\documentclass
[11pt, a4paper]{amsart}
\usepackage[utf8]{inputenc}
\usepackage{amssymb,amscd}
\usepackage{thmtools}

\usepackage[textsize=tiny]{todonotes}
\usepackage{hyperref}
\usepackage[shortlabels]{enumitem}
\usepackage{float} 
\usepackage{cleveref}   
\usepackage{stmaryrd}

\makeatletter
\@namedef{subjclassname@2020}{\textup{2020} Mathematics Subject Classification}
\makeatother


\newtheoremstyle{alstandard}{7pt}{3pt}{\rm}{}{\scshape}{:}{0.5em}{}
\theoremstyle{alstandard}
\swapnumbers
\declaretheorem[name=Theorem]{theorem}
\numberwithin{theorem}{section}
\declaretheorem[sibling=theorem, name=Lemma]{lemma}
\declaretheorem[sibling=theorem, name=Proposition]{prop}

\declaretheorem[sibling=theorem, name=Remark]{remark}

\declaretheorem[sibling=theorem, name=Claim]{claim}

\declaretheorem[sibling=theorem, name=Question]{question}

\newcommand{\R}{\mathbb{R}}
\newcommand{\C}{\mathbb{C}}

\newcommand{\K}{\mathbb K}

\DeclareMathOperator{\im}{im}

\DeclareMathOperator{\rk}{rank}

\DeclareMathOperator{\GL}{GL}

\usepackage{pifont}
\newcommand{\comment}[1]{}
\begin{document}
	\title[Chow pencils]{A linear-time algorithm for Chow decomposition
	} 
		
	\author[Blomenhofer]{Alexander Taveira Blomenhofer}
	\address{University of Copenhagen, Lyngbyvej 2, 2100 Copenhagen \o, Denmark}
	\email{atb@math.ku.dk}
	
	\author[Lovitz]{Benjamin Lovitz}
	\address{Concordia University, 1455 Blvd. De Maisonneuve Ouest, Montreal, Quebec, Canada}
	\email{benjamin.lovitz@concordia.ca, Benjamin.Lovitz@gmail.com}

	\begin{abstract} 
		We propose a linear-time algorithm to compute low-rank Chow decompositions. Our algorithm can decompose concise symmetric $ 3 $-tensors in $n$ variables of Chow rank $n/3$. The algorithm is pencil-based, hence it relies on generalized eigenvalue computations. We also develop sub-quadratic time algorithms for higher order Chow decompositions, and Chow decompositions of $ 3 $-tensors into products of linear forms which do not lie on the generic orbit. In particular, we obtain a sub-quadratic-time algorithm for decomposing a symmetric 3-tensor into a linear combination of W-tensors.
	\end{abstract}

	\maketitle
	\raggedbottom
	\section{Introduction}
Let $T \in (\K^n)^{\otimes 3}$ be a 3-tensor over $\K\in \{\R, \C\}$, and let $ \mathcal{L}_T \subseteq (\K^n)^{\otimes 2}$ be the space of all matrices which arise from contraction of $ T $ by some vector $ f\in \mathbb \K^n $ along the first index. 
We define the \emph{contraction variety} $X_T$ to be the set of all vectors $ f\in \mathbb \K^n $ such that the contraction $ T_f $ of $ T $ along $ f $ has rank strictly smaller than a generic contraction. For instance, if general elements of $ \mathcal{L}_T $ have full rank, then we have
\begin{align*}
	X_T = \{f\in \mathbb \K^n\mid \det(T_f) = 0\}. 
\end{align*}
In the low-rank setting, the contraction variety often carries information that allows one to compute the minimum rank decomposition of a tensor and certify its uniqueness. For instance, if $ T\in \K^{n\times n\times n} $ is a concise symmetric tensor of symmetric rank at most $ n $, then its contraction variety is a union of $ n $ hyperplanes. Indeed, if $ T $ has a decomposition $ T = \sum_{i = 1}^{n} a_i^{\otimes 3} $, then by conciseness, $ a_1,\ldots,a_n $ form a basis of $ \K^n $. We will later see that then
\begin{align*}
	X_T = \bigcup_{i = 1}^n \langle a_i \rangle^{\perp}. 
\end{align*}
Here and throughout, we take the orthogonal complement with respect to the bilinear pairing $\langle a, b \rangle = a^T b$, which is an inner product for $ \K = \R $. Hence, the contractions $ T\mapsto T_f $ that make $ T_f $ singular are precisely those for which $ f $ is orthogonal to one of $ a_1,\ldots,a_n $. This fact can be turned into a decomposition algorithm via matrix pencils. For symmetric rank, it is classically known how to compute the minimum rank decomposition under the given assumptions. In fact, algorithms for concise tensors of symmetric rank $ n $ have been discovered and rediscovered so many times that it is difficult to give correct attribution to a single reference. See the discussion under \emph{related work}. 
 
Surprisingly, we find that a very similar approach also works for tensors of low Chow rank. 
Recall that the \textit{Chow rank} is the minimum decomposition of a symmetric tensor $T$ into terms of the form $a_{i1} a_{i2} a_{i3}$, where $$ a_{i1} a_{i2} a_{i3}:=\frac{1}{3!} \sum_{\sigma \in \mathfrak{S}_3} a_{i\sigma(1)}\otimes a_{i \sigma(2)}\otimes a_{i \sigma(3)} $$ denotes the symmetrization of the tensor product $ a_{i1} \otimes a_{i2} \otimes a_{i3} $. If we view $T$ as a homogeneous polynomial, then this decomposes $T$ into a sum of products of linear forms.

We will see that the contraction variety of a concise symmetric order $ 3 $ tensor of Chow rank $ n/3 $ is also a union of hyperplanes. Precisely, if the rank-1 factors are given by the vectors $ \{a_{i1}a_{i2}a_{i3}\}_{i=1,\ldots,n/3} $ as in \eqref{eq:chow-decomposition-order-3}, then $\{a_{ij}\}_{i=1,\dots, n/3, j=1,2,3}$ forms a basis for $\K^n$, and we have
\begin{align*}
	X_T = \bigcup_{i=1}^r \langle a_{i1} \rangle^{\perp} \cup \langle a_{i2} \rangle^{\perp} \cup \langle a_{i3} \rangle^{\perp}.
\end{align*}
This fact will allow us to give a linear-time, pencil-based algorithm to compute Chow decompositions.
\subsubsection{Organization and overview of results}
Our main result is the following:
\begin{theorem}
Let $n=3r$, and let $ T\in \K^{n\times n\times n} $ be a concise symmetric 3-tensor of Chow rank $r$. Then the Chow rank decomposition of $T$ is unique, and it can be recovered in time $\mathcal{O}(n^3)$.
\end{theorem}
Note that the runtime is linear in the size $n^3$ of the tensor. This result is presented in \Cref{sec:core}. The rank assumptions, together with conciseness, imply that each Chow rank 1 term lies on the generic orbit of the Chow variety, which is $ \GL_n e_1e_2e_3 $. In \Cref{sec:cubic-nonindependent}, we examine what can be said if one allows non-generic orbits, i.e., if the Chow rank 1 components may be chosen from any of the $ 4 $ orbits in the Chow variety (namely, the orbits of $e_1^3, e_1 e_2 e_3, e_1^2 e_2,$ and $e_1 e_2 (e_1+e_2)$). In particular, we obtain the following:
\begin{theorem}\label{thm:nongeneric-orbits-intro}
Let $T=\sum_{i=1}^r a_{i1}a_{i2}a_{i3}$ be a symmetric tensor for which the spaces $U_i=\langle a_{i1},a_{i2},a_{i3}\rangle$ are in direct sum. Then a minimum Chow rank decomposition of $T$ can be obtained in  time $\mathcal{O}(n^4)$.
\end{theorem}
The Chow rank decomposition is unique over $\K=\R$, and has a controlled non-uniqueness over $\K=\C$ resulting from the fact that $e_1^3+e_2^3$ has Chow rank 1. See~\Cref{sec:cubic-nonindependent} for more details. In particular, \Cref{thm:nongeneric-orbits-intro} gives an $\mathcal{O}(n^4)$-time algorithm to decompose a symmetric tensor into W-tensors. These are Chow rank 1 tensors of the form $a_{i1}^2 a_{i2}$.

In \Cref{sec:higher-order-chow}, we discuss a generalization to symmetric tensors $ T $ of odd order $ 2d+1 $, that is, $ T\in S^{2d+1}(\mathbb C^n) $. We show that if $ T $ is general of Chow rank $ r= \mathcal{O}(n^d) $, then there is an efficient algorithm to compute the unique minimum Chow rank decomposition. Note that pencil-based methods admit a similar odd-order generalization for symmetric rank, which also allows to go up to rank $ \mathcal{O}(n^d) $. The $ \mathcal{O} $-notation hides constants depending on $ d $. We leave open the question of generalizing our construction to even order.

While our main focus is on Chow decompositions, we expect our methodology to be of independent interest. Our approach can in principle be extended to decomposition problems for which the contraction variety is an arrangement of planes of constant codimension. 

\subsubsection{Model of computation} Here and throughout, we work in the real (or complex) model of computation, and assume access to an oracle to diagonalize a $k \times k$ diagonalizable matrix in time $\mathcal{O}(k^3)$. This is a standard model of computation for many tensor decomposition algorithms already present in the literature, e.g.~\cite{johnston2023computing,kothari2024overcomplete,koiran2025efficient}.

\subsubsection{Related work} 
As mentioned above, methods for low-rank tensor decomposition have been studied from an ample variety of angles. For symmetric tensor rank, algorithms exist based on either simultaneous diagonalization, tensor power iteration \cite{Anandkumar_Ge_Hsu_Kakade_Telgarsky_2012}, semidefinite programming \cite{Ma_Shi_Steurer_2016}, the eigenvalue method \cite{Telen_Vannieuwenhoven_2022} and, most importantly for this work, pencils \cite{Leurgans_Ross_Abel_1993}. The earliest algorithmic results of this flavour appear to be due to Sylvester \cite{Sylvester_1904} and Prony \cite{Prony_1795} in 1795. See \cite{Mourrain_2018} for a modern exposition of Prony's work. A similar algorithm based on simultaneous matrix diagonalization was introduced by Harshman~\cite{Har72} and later by Leurgans, Ross, and Abel~\cite{Leurgans_Ross_Abel_1993} (and then rediscovered several times thereafter). Matrix pencils have also been studied from the geometric viewpoint, e.g., via the variety of complete collineations \cite{Manivel_Michalek_etal_2024}, \cite{Gesmundo_Keneshlou_2025}.

At the same time, comparatively little is known for the Chow variety. They are overall one of the least understood standard models of tensor decomposition (see  \cite[p. 2]{Torrance_Vannieuwenhoven_2022}). The identifiability of Chow decompositions has only recently been addressed (\cite{Torrance_Vannieuwenhoven_2022}, \cite{Blomenhofer_Casarotti_2023}). Prior work of Kayal and Saha~\cite{kayal2019reconstruction} gave an $n^c$-time algorithm for Chow decomposition in the black-box model of computation for an uncontrolled constant $c$. By contrast, the exponents for our algorithms are tightly controlled, and we can handle similar rank constraints. For example, for order-3 tensors our algorithm can decompose tensors up to Chow rank $n/3$, while that of~\cite{kayal2019reconstruction} handles Chow rank $n/9$. More recently, \cite{garg2020learning} (see also \cite{chandra2023learning}) gave $\text{poly}(n)$-time algorithms for decomposing homogeneous polynomials into powers of low-degree polynomials in the black-box model.

\subsubsection{Applications} Chow decompostions of cubic forms correspond to depth-three $\Sigma \Pi \Sigma$ arithmetic circuits, and the Chow rank of a form determines the minimum fan-in of any such circuit that computes it. See e.g. \cite{karnin2009reconstruction} and the references therein.
Chow decompositions are also relevant in quantum information theory, as Chow rank-1 terms are symmetrizations of pure product states, and include many commonly used examples such as W-states and GHZ-states. In particular, our algorithm from \Cref{sec:cubic-nonindependent} can be used to decompose a given tripartite quantum state into a superposition of W-states.

		\section{Chow decompositions for third order tensors}\label{sec:core}

Let $\K \in \{\R, \C\}$. A Chow decomposition of a symmetric order-$ 3 $ tensor $ T\in S^3(\K^n) $ has the form
\begin{align}\label{eq:chow-decomposition-order-3}
	T = \sum_{i = 1}^{r} a_{i1} a_{i2} a_{i3}
\end{align}
where for vectors $v_1,\dots, v_d \in \K^d$ we let $v_1 \cdots v_d := \frac{1}{d!} \sum_{\sigma \in \mathfrak{S}_d} v_{\sigma(1)} \otimes \cdots \otimes v_{\sigma(d)} $. If $ r $ is minimal, then~\eqref{eq:chow-decomposition-order-3} is called a Chow-rank decomposition.

We consider the case where the set $ \{a_{i1}, a_{i2}, a_{i3} \}_{i = 1,\ldots,r} $ of all vectors occurring in the decomposition is linearly independent. Note that this holds whenever $T$ is a concise tensor of Chow rank $r=n/3$. In the following, we determine the contraction variety. 
\begin{prop}\label{prop:contraction-variety-chow}
	Let $ T $ be a concise tensor with a Chow decomposition as above and where $\{a_{i1}, a_{i2}, a_{i3} \}_{i = 1,\ldots,r} $ is linearly independent. In particular, $ n = 3r $. Then, 
	\begin{align*}
		X_T = \bigcup_{i=1}^r \langle a_{i1} \rangle^{\perp} \cup \langle a_{i2} \rangle^{\perp} \cup \langle a_{i3} \rangle^{\perp}.
	\end{align*}
\end{prop}
\begin{proof}
	For $ f\in \C^n $, we may write $ T_f = \sum_{i = 1}^{r} M_{if} $, where $ M_{if} $ is the contraction of $ a_{i1} a_{i2} a_{i3} $ along $ f^T$. Note that the images of $ M_{if} $ are contained in $ U_i = \langle a_{i1},a_{i2},a_{i3} \rangle $ and the sum of the spaces $ U_i $ is direct. Therefore, it holds that $ \im T_f = \bigoplus_{i=1}^r \im M_{if} $. Consequentially, $ T_f $ is rank-deficient if and only if at least one of the $ M_{if} $ has rank less than $ 3 $. 
	
	(As an aside, note that for generic $ f $, all $ M_{if} $ have rank $ 3 $ and $ T_{f} $ has rank $ 3r $, which equals $ n $ by conciseness). 
	
	Let us find the set of $f$ for which $ M_{if} $ has rank less than $ 3 $. For this purpose, we can assume without loss of generality that $a_{ij}=e_j$ for $j=1,2,3$, and that $f=\alpha e_1 + \beta e_2 + \gamma e_3 + g$ for some $g \in \im(M_{if})^{\perp}$. Then $3 M_{if}=\alpha e_2 e_3 + \beta e_1 e_3 + \gamma e_1 e_2$, which has determinant $\frac{1}{4}\alpha\beta \gamma$. So $M_{if}$ is rank deficient if and only if $ f\in \langle a_{i1} \rangle^{\perp} \cup \langle a_{i2} \rangle^{\perp} \cup \langle a_{i3} \rangle^{\perp} $. This completes the proof.
\end{proof}

\Cref{prop:contraction-variety-chow} shows that the contraction variety of a concise tensor with a low-rank Chow decomposition is a union of codimension-1 subspaces. Therefore, we can proceed with a pencil-based algorithm and intersect $ X_T $ with a line, just like in the classical pencil-based algorithm for tensor rank. 
By \Cref{prop:contraction-variety-chow}, the generalized eigenvalue problem
\begin{align*}
	\det(T_f - \lambda T_g) = 0
\end{align*}
will have $ 3r $ solutions, which we label $ \lambda_{i 1},\lambda_{i 2},\lambda_{i 3} $, where $ i = 1,\ldots,r $. Let the labeling be chosen such that $ f-\lambda_{ik}g\in \langle a_{ik} \rangle^{\perp} $. As the solutions correspond to intersection points of the generic line $f-\lambda g$ 
with $ X_T $, the $ 3r$ solutions will be pairwise distinct. The corresponding eigenvectors we label accordingly as $ x_{i1}, x_{i2}, x_{i3} $. The eigenspaces are one-dimensional, as there are $ 3r $ pairwise distinct eigenvalues. 
Denote $ h_{ik} := f-\lambda_{ik}g \in \C^n$.

Note that $x_{i k} \perp U_j$ for $j \neq i$. Indeed, $M_{jh_{ik}} x_{i k}=0$, and $M_{jh_{ik}}$ is a symmetric matrix with $\im(M_{jh_{ik}})=U_j$.

\begin{remark}
	Even when $ \K=\R $, we always consider $ U_i $ and $ \im M_{if} $ as complex subspaces. While the matrices $ T_f $ and $ T_g $ are symmetric, they may not be positive definite (for instance, $ (e_1e_2e_3)_f $ has zeros on the diagonal). Therefore, the generalized eigenvalue problem can have complex solutions. In other words, for real $ f, g\in \R^n $, the intersection points of the line $ \lambda \mapsto f + \lambda g $ with $ X_T $ may be complex. 
	Therefore, we will always perform algorithmic operations over $ \K = \C $, unless explicitly stated otherwise. 
\end{remark}

In the following Proposition, we show a criterion to match eigenvectors together, which belong to the same Chow rank 1 term. 

\begin{prop}\label{prop:source-trick-for-dim-Ui-3}
Let $\ell \in \K^n$ be generic. Then $x_{i k}^T T_{\ell} x_{i' k'} \neq 0$ if and only if $i=i'$.
\end{prop}
\begin{proof}
If $i \neq i'$ then $x_{i k}^T T_{\ell} x_{i' k'} = 0$, since $x_{i k} \perp U_{i'}$ for $i' \neq i$. If $i=i'$, then it suffices to exhibit a single choice for which $x_{i k'}^T T_{\ell} x_{i k} \neq 0$, as this is an open condition on $\ell$. By symmetry, it suffices to prove that $x_{1 1}^T T_{\ell} x_{1 k} \neq 0$ for $k=1,2$. Choose $\ell$ so that $\langle \ell, a_{ik}\rangle = \delta_{i=1, k=1}$. Then
\begin{align*}
	x_{1 1}^T T_{\ell} x_{1 k}=x_{1 1}^T \sum_{i} M_{i\ell} x_{1k} = x_{1 1}^T M_{1 \ell} x_{1k}= \alpha_k \langle x_{11},a_{12} \rangle \langle x_{1k}, a_{13} \rangle,
\end{align*}
where $\alpha_1=\frac{1}{3}$ and $\alpha_2=\frac{1}{3!}$. Suppose toward contradiction that one of the inner products on the righthand side is zero, e.g. the first. Note that
\begin{align}\label{eq:M_h}
	M_{h_{11}}=\frac{1}{3!}(a_{11}v^T+ v a_{11}^T),
\end{align}
where $v=\langle h_{11}, a_{12}\rangle a_{13}+\langle h_{11}, a_{13}\rangle a_{12}$. Since there are no repeated eigenvalues, the coefficients of $a_{13}$ and $a_{12}$ are both nonzero. Since $M_{h_{11}} x_{11}=0$, we have
\begin{align}\label{eq:containment_verified}
x_{11} \in \langle a_{11} \rangle^{\perp} \cap \langle v \rangle^{\perp}.
\end{align}
But if it furthermore holds that $\langle x_{11},a_{12} \rangle=0$, then $x_{11} \perp U_1$, so $x_{11}=0$, a contradiction. This completes the proof.
\end{proof}

Thus, by examining the inner products $x_{i' k'}^T T_{\ell} x_{i k}$ we can partition the eigenvectors into triples $\{x_{i 1}, x_{i 2}, x_{i 3}\}$ and obtain a basis $\{T_f x_{i1},T_f x_{i2},T_f x_{i3}\}$ for $U_i$. To see that they form a basis, note that $T_f$ is a full-rank matrix, $x_{i1},x_{i2},x_{i3}$ are linearly independent, and $T_f x_{ik} \in U_i$ for $k=1,2,3$.

We now claim that $\langle x_{ik}\rangle^{\perp} \cap \langle h_{ik} \rangle^{\perp} \cap U_i = \langle a_{ik} \rangle$. By symmetry it suffices to prove this for $i=k=1$. The containment $\supseteq$ is easily verified from Proposition~\ref{prop:contraction-variety-chow} and~\eqref{eq:containment_verified}. For the containment $\subseteq$, it is evident from the formula~\eqref{eq:M_h} that
\begin{align*}
	\langle x_{1 1}\rangle^{\perp} \cap U_1 = \langle a_{1 1}, v \rangle.
\end{align*}
Since clearly $v \notin \langle h_{11} \rangle^{\perp}$, the claim is proven.

The above shows that we can obtain non-zero vectors $a_{i k}'$ in $\langle a_{i k} \rangle$. We know that $T=\sum_{i=1}^r c_i a_{i1}' a_{i2}' a_{i3}'$ for some scalars $c_i$, and it remains only to compute these $c_i$. To do this, we solve the system $T_f \; x_{i 1}= c_i (a_{i_1}' a_{i_2}' a_{i_3}')_f \; x_{i 1}$ for $c_i \in \K$. Here, we use again that $ x_{i1} \perp U_{i'} $ for $ i' \ne i $. Finally, we output $\{c_i a_{i_1}' a_{i_2}' a_{i_3}' : i =1,\dots, r\}$.

\medskip\noindent
We have proved the following theorem.

\begin{theorem}\label{thm:chow}
	Suppose that $T=\sum_{i=1}^m a_{i1}a_{i2} a_{i3}$ with $\{a_{ik}\}_{i\in [r],k=1,2,3}$ linearly independent. Then the following algorithm produces the terms $a_{i1} a_{i2} a_{i3}$:
	\begin{enumerate}
		\item Sample generic $ f, g, \ell \in \K^n $. 
		\item Find the $ 3r $ solutions $ \lambda_{ik}\in \C, x_{ik} \in \mathbb C^n $ to the generalized eigenvalue problem $ \det(T_f - \lambda T_g) = 0 $, where $ i = 1,\ldots,r $ and $ k = 1,2,3 $. 
		\item Set $ h_{ik} := f - \lambda_{ik} g $.
		\item For each vector $x_{i k}$, find the vectors $x_{i' k'}$ for which $x_{i' k'}^T T_{\ell} x_{i k} \neq 0$. This partitions the vectors into triples $\{x_{i 1}, x_{i 2}, x_{i 3}\}$.
		\item Let $U_i=\langle T_f x_{i1}, T_f x_{i2}, T_f x_{i 3} \rangle$.
		\item Compute nonzero vectors
		\begin{align*}
			a_{i k}' \in \langle x_{i k} \rangle^{\perp} \cap \langle h_{i k} \rangle^{\perp} \cap U_i.
		\end{align*}
		for $k=1,2,3$.
		\item Solve the system $T_f \; x_{i1}= c_i (a_{i1}' a_{i2}' a_{i3}')_f \; x_{i1}$ for $c_i \in \K$, and output $\{c_i a_{i1}' a_{i2}' a_{i3}' : i =1,\dots,r\}$.
	\end{enumerate}
\end{theorem}

\subsection{Runtime}
The algorithm from \Cref{thm:chow} has runtime $ \mathcal{O}(n^3)$. Computing $ T_f, T_g,$ and $T_{\ell}$ takes time $ \mathcal{O}(n^3) $. Solving the generalized eigenvalue problem takes time $\mathcal{O}(n^3) $. In Step (4), it takes time $ \mathcal{O}(n^3)$ to compute the vectors $T_{\ell} x_{i' k'}$, and then time $\mathcal{O}(n)$ to compute each of the $\mathcal{O}(n^2)$ inner products $x_{i k}^T T_{\ell} x_{i' k'}$,  giving a total runtime of $ \mathcal{O}(n^3)$. Step (6) can be computed in time $\mathcal{O}(n^3)$. In Step (7), for each $i$ we need to compute $T_f \; x_{i1}$, compute $(a_{i1}' a_{i2}' a_{i3}')_f \; x_{i1}$, and solve for $c_i$. Each of these steps takes time $\mathcal{O}(n^2)$, so doing this for $i=1,\dots, n$ takes time $\mathcal{O}(n^3)$ in total.

	\section{Chow decompositions for higher order tensors}\label{sec:higher-order-chow}
In this section, we develop an algorithm to compute Chow decompositions for tensors of odd order $ 2d+1 $. In the case $ d = 1 $, we recover exactly the algorithm from the previous section. 

Let us start from a general tensor $ T $ of Chow rank $ r $, say with minimum rank decomposition
\begin{align*}
	T = \sum_{i = 1}^{r} a_{i1}\cdot \cdots \cdot a_{i(2d+1)} 
\end{align*}
and $ a_{ij} $, where $ i=1,\ldots,r $, $ j = 1,\ldots,2d+1 $ in general position. After contracting by a vector $ f\in \C^n $, we obtain a symmetric tensor of order $ 2d $, which may be reshaped into a matrix $ T_f\colon S^d(\mathbb C^n) \to S^d(\mathbb C^n) $. 
Clearly, we have that $ T_f = \sum_{i = 1}^{r} M_{if} $, where $ M_{if} $ denotes the matrix reshaping of the contraction of $ a_{i1}\cdot \cdots \cdot a_{i(2d+1)} $ by $ f $. 

First, we need to establish that the images of the matrices $ M_{1f},\ldots,M_{rf} $ will be in direct sum. Fortunately, this is guaranteed by a recent work, see \cite{Blomenhofer_Casarotti_2023}. 

\begin{lemma}\label{lem:higher-order-direct-sum}
	For $ b = (b_1,\ldots,b_{2d+1})\in (\mathbb C^n)^{2d+1} $, define  $ L(b) $ as the subspace of $ S^d(\mathbb C^n) $ spanned by all vectors
	\begin{align*}
		b_{i_1}\cdot \cdots \cdot b_{i_d} \in  S^d(\mathbb C^n), 
	\end{align*} 
	where $ i_1,\ldots,i_d $ are $ d $ pairwise distinct indices chosen from $ \{1,\ldots,2d+1\} $. 
	\noindent
	Then, a sum of $ r $ general spaces $ L(b^{(1)})  + \ldots + L(b^{(r)}) $ is direct, as long as 
	\begin{align*}
		r \le\dfrac{\binom{n+d-1}{d}}{\binom{2d+1}{d}} - \binom{2d+1}{d}. 
	\end{align*}
	In particular, the symmetric matrices $ M_f(b^{(1)}),\dots, M_f(b^{(r)})\colon S^d(\mathbb C^n) \to S^d(\mathbb C^n)  $, given by 
	\begin{align*}
		M_{f}(b^{(q)}) = \frac{d!^2}{2(2d+1)!} \sum_{I\dot{\cup} J \dot{\cup} \{k\} = [2d+1] } \langle f, b^{(q)}_k \rangle\cdot (b^{(q)}_I)(b^{(q)}_J)^{T} + \langle f, b^{(q)}_k \rangle \cdot (b^{(q)}_J)(b^{(q)}_I)^{T}, 
	\end{align*}
	have images in direct sum. Here, $b^{(q)}_I:=b^{(q)}_{i_1}\cdots b^{(q)}_{i_d}$, and the sum runs over all partitions of $ [2d+1] $ into disjoint sets, where $ I $ and $ J $ both contain $ d $ elements. 
\end{lemma}
\begin{proof}
	First, note that the image of the matrix $ M_f(b) $ is contained in $ L(b) $ for all $ f $ and all $ b \in (\mathbb C^n)^{2d+1} $. Next, note that the class of all spaces $ L(b) $ is $ \GL_n $ invariant, since $ U^{\otimes d} L(b) =  L(Ub_1,\ldots,Ub_{2d+1}) $ for each matrix $ U\in \GL_n $. Additionally, $ S^d(\mathbb C^n) $ is an irreducible $ \GL_n $-module, by Schur-Weyl duality. 
	This means that the conditions of \cite[Theorem 3.4]{Blomenhofer_Casarotti_2023} are satisfied (if we take $ T_b := L(b) $ in \cite[Theorem 3.4]{Blomenhofer_Casarotti_2023} and $ X $ as the dense open subset of $(\mathbb C^n)^{2d+1} $ where $ L(b) $ has constant dimension).  
	We obtain that for general $ b^{(1)},\ldots,b^{(r)}\in (\mathbb C^n)^{2d+1} $, the spaces $ L(b^{(1)}),\ldots,L(b^{(r)}) $ are in direct sum. In particular, the images of $ M(b^{(1)}),\ldots,M(b^{(r)}) $ are in direct sum. 
\end{proof}

\begin{remark}\label{rem:generic-dim-of-L(b)}
	In \Cref{lem:higher-order-direct-sum}, the dimension of a space $ L(b) $ is upper bounded by $ \binom{2d+1}{d} $ for each $ b \in (\mathbb C^n)^{2d+1} $. Additionally, if $ 2d+1\le n $, then we can easily see that this bound is an equality for generic $ b $. Indeed, after a change of coordinates, it suffices to check this statement for the standard basis $ b = e = (e_1,\ldots,e_n) $. Then, the expressions $ b_I $ correspond bijectively to the square-free monomials $ x_I = \prod_{i\in I} x_i $ in variables $ x_1,\ldots,x_{2d+1} $ of degree $ d $, which are of course linearly independent. 
	Using that, it is immediate to see that the image of $ M_{f}(e) $ equals $ \langle e_{I} \mid I\subseteq [2d+1], |I| = d \rangle $ for general $ f $ and thus the rank of $ M_f(e) $ is $ \binom{2d+1}{d} $. 
\end{remark}

Let us now consider again the contraction variety. 
From \Cref{lem:higher-order-direct-sum}, we obtain that the rank of $ T_f $ drops below its generic rank, if and only if one of $ M_{1f},\ldots,M_{rf} $ has rank less than $ \binom{2d+1}{d} $. In the next step, we will examine when the latter can happen. Our goal is to prove the following statement. 

\begin{theorem}\label{thm:contraction-variety-higher-chow}
	Let $ T = \sum_{i = 1}^{r} a_{i1}\cdot \cdots \cdot a_{i(2d+1)} $ be a generic symmetric tensor of Chow-rank $ r $, where $ r \le \dfrac{\binom{n+d-1}{d}}{\binom{2d+1}{d}} - \binom{2d+1}{d}$ and $ 2d+1\le n $. 
	Then, its contraction variety is given by
	\begin{align*}
		X_T = \bigcup_{i = 1}^{r} \langle a_{i1} \rangle^{\perp} \cup \ldots \cup \langle a_{i(2d+1)} \rangle^{\perp}.
	\end{align*}
\end{theorem}
\begin{proof}
	By \Cref{lem:higher-order-direct-sum} and \Cref{rem:generic-dim-of-L(b)}, the matrices $ M_{if}=M_f(a_{i1},\dots, a_{i(2d+1)}) $ have rank equal to $ \binom{2d+1}{d} $, if one considers a contraction by a generic $ f\in \mathbb C^n $. By \Cref{lem:higher-order-direct-sum}, the rank of a generic contraction $ T_f $ therefore equals $ r\binom{2d+1}{d} $ under the given assumptions $ 2d+1\le n $ and $ r \le\dfrac{\binom{n+d-1}{d}}{\binom{2d+1}{d}} - \binom{2d+1}{d} $. In this regime, the contraction variety of a general Chow tensor of rank $ r $ is thus given by
	\begin{align*}
		X_T = \{f\in \mathbb C^n\mid \rk T_f < r\binom{2d+1}{d} \}
	\end{align*}
	By \Cref{lem:higher-order-direct-sum}, we have that the images of $ M_{1f},\ldots,M_{rf} $ are in direct sum. Therefore, 
	\begin{align*}
		X_T = \{f\in \mathbb C^n\mid \exists i\colon \rk M_{if} < \binom{2d+1}{d} \}. 
	\end{align*}
	We will show in the following \Cref{lem:rank-drops-of-Mb} that the rank of $ M_{if} $ drops if and only if $ f $ is orthogonal to one of the vectors $ a_{i1},\ldots,a_{i(2d+1)} $. Thus, we conclude with the desired result. 
\end{proof}

\begin{lemma}\label{lem:rank-drops-of-Mb}
	Let $ b_1,\ldots,b_{2d+1}\in \mathbb C^n $ be general and assume that $ 2d+1\le n $. Consider the Chow rank 1 tensor $ M(b) := b_1\cdots b_{2d+1} \in S^{2d+1}(\mathbb C^n) $ and for each $ f\in \mathbb C^n $ the matrices $ M_f(b) $ from \Cref{lem:higher-order-direct-sum}. Then, 
	\begin{align*}
		\rk M_f(b) < \binom{2d+1}{d} \Longleftrightarrow f \in \langle b_{1} \rangle^{\perp} \cup \ldots \cup \langle b_{2d+1} \rangle^{\perp}.
	\end{align*}
\end{lemma}
\begin{proof}
By using the $ \GL_n $ action, we may w.l.o.g. assume that $ b $ consists of standard basis vectors, so that $ b = e:= (e_1,\ldots,e_{2d+1}) $.

We begin by proving the direction ``$\implies$." Let $\alpha_i=\langle e_i,f\rangle$. Clearly, acting with the torus $ (\C^{\times})^{2d+1} \subseteq (\C^{\times})^{n} $ on $e$ cannot change the rank of $ M_f(e) $. If we let $\beta:=(\beta_1,\dots, \beta_{2d+1})$, then
\begin{align*}
M_f (\beta \cdot e)=\alpha_1\beta_2\cdots \beta_{2d+1} \frac{e}{e_1}  + \ldots + {\beta_1\cdots \beta_{2d}\alpha_{2d+1}}\frac{e}{{e_{2d+1}}},
\end{align*}
where $e/e_i:=e_1\cdots e_{\hat{i}} \cdots e_{2d+1}$ is the Chow rank 1 tensor of order $ 2d $ where the term $ e_{i} $ is omitted. Choosing $ \beta_i = \frac{\alpha_i}{\sqrt[2d]{\alpha_1\cdots \alpha_{2d+1}}} $, we obtain $M_{1,\dots, 1}(e)=\frac{e}{e_1}+\dots+\frac{e}{e_{2d+1}}$.
	Hence, $ M_f(e) $ is rank deficient if and only if $ M_{(1,\ldots,1)}(e) $ is rank deficient. Since a general vector $ f\in \mathbb C^n $ will satisfy the open condition $ \langle e_i, f \rangle \ne 0$, we conclude that $ \rk M_f(b) = \binom{2d+1}{d} $, if all $ \langle b_i, f \rangle $ are nonzero. This shows the direction ``$\implies$''.
 
	\medskip\noindent
	Let us now show the direction ``$\impliedby$''. Let $\Pi$ be the matrix for which $\Pi_{L(e)}=I_{L(e)}$ and $\Pi_{L(e)^{\perp}}=0$, where $L(e)^{\perp}$ is the orthogonal complement to $L(b)$ with respect to the bilinear pairing $\langle x,y\rangle=x^Ty$ (or the sesquilinear form; they are equivalent since $L(e)$ has a real basis).
	 Note that $M_{\ell}(e) \Pi = M_{\ell}(e)$ for any $\ell \in \C^n$. Let us view $M_{\ell}(e) \Pi$ as a map on $L(e)$, which is nonsingular for generic $\ell \in \C^n$.  Since the set of singular matrices has codimension 1, for generic $f,g \in \K^n$ the (complex) subspace spanned by $M_f(e) \Pi , M_g(e) \Pi$ must contain a matrix $M_{\ell}(e) \Pi$ whose image is properly contained in $L(e)$, i.e., a matrix which is singular on $ L(e) $. Hence, $M_{\ell}(e)$ has subgeneric rank.

	In particular, there exists some $ \lambda $ such that $ f-\lambda g\in X_{M(b)} $. By the already shown direction ``$\implies$'', we know that there exists some $ i\in \{1,\ldots,2d+1\} $ such that $ f-\lambda g \in \langle b_i \rangle^{\perp} $. Fix such $ i $. Note that $ f-\lambda g\notin \langle b_j \rangle^{\perp}  $ for $ j\ne i $, since a general line cannot intersect the codimension 2 subspace $ \langle b_i \rangle^{\perp} \cap \langle b_j \rangle^{\perp}$. Furthermore, there can only be precisely one $ \lambda $ such that  $ f-\lambda g\in \langle b_i \rangle^{\perp} $, since the degree of a hyperplane is $ 1 $.  It is therefore justified that we denote this unique $ \lambda $ by $ \lambda_i $. 
	
	We thus showed that the intersection $ B = X_{M(b)} \cap  \langle b_i \rangle^{\perp} \ne \{0\} $ is nontrivial. However, the set $ B $ has to be of codimension $ 1 $, since otherwise the general line $ \lambda \mapsto f - \lambda g $ could not intersect $ B $. Since $ \langle b_i \rangle^{\perp} $ is irreducible, this shows that $ \langle b_i \rangle^{\perp}\subseteq X_{M(b)} $. However, the matrix $ M(b) $ and thus also the contraction variety $ X_{M(b)} $ are invariant under permutations $ b \mapsto (b_{\sigma(1)},\ldots,b_{\sigma(2d+1)}) $, where $ \sigma\in \mathfrak{S}_{2d+1} $. Therefore, by symmetry, we conclude that $ \langle b_{1} \rangle^{\perp} \cup \ldots \cup \langle b_{2d+1} \rangle^{\perp} \subseteq X_{M(b)} $. This shows ``$\impliedby$''.
\end{proof}

\begin{remark}
	From the proof of \Cref{lem:rank-drops-of-Mb}, we conclude that the general line $ f - \lambda g $ intersects $ X_{M(b)} $ precisely in the $ 2d+1 $ points $ \lambda_i $ of the line for which $ f - \lambda_i g  \in \langle b_i \rangle^{\perp}$. Since $ \dim V = \binom{2d+1}{d} $, this means that by symmetry, each generalized eigenvalue has algebraic multiplicity $ \binom{2d+1}{d}/(2d+1) $. 
	
	We will show in the following that the geometric multiplicity is equal to the algebraic multiplicity for points in the generic orbit. This is nontrivial and needs a specific combinatorial argument, see \Cref{prop:geom_mult}. In fact, for non-generic orbits, it is not always true that the algebraic and geometric multiplicity agree. We will see a counterexample in \Cref{sec:cubic-nonindependent}: Already for $ d = 1 $, on the orbit of the W-tensor $ b_1^2 b_2 $, the algebraic multiplicity is $ 2 $ and the geometric multiplicity is $ 1 $. 
\end{remark}


\begin{prop}\label{prop:geom_mult}
	Let $T=e_1\cdots e_{2d+1}$ and let $\partial=e_1+\dots+e_{2d}$. Then $T_{\partial}$ has rank $2\binom{2d}{d-1}$.
\end{prop}
Note that this proposition shows that the kernel has dimension $\binom{2d+1}{d}-2\binom{2d}{d-1}=\binom{2d+1}{d}/(2d+1)$, proving that the geometric multiplicity is as expected.

The proof requires a well-known fact about Johnson graphs. For positive integers $k<n$, the \textit{Johnson graph} $J(n,k)$ is the undirected graph with vertices given by subsets of $[n]$ of size $k$, with two vertices adjacent if and only if the corresponding subsets intersect in exactly $k-1$ elements. Let $A(n,k)$ be the adjacency matrix of $J(n,k)$.
\begin{lemma}[Theorem 9.1.2 in \cite{brouwer2011distance}]\label{lemma:johnson}
	The eigenvalues of $A(n,k)$ are given by $\alpha_{n,k}(j):=(k-j)(n-k-j)-j$ with multiplicity $\binom{n}{j}-\binom{n}{j-1}$ for $j=0,\dots, \min\{k,n-k\}$.
\end{lemma}

\noindent
Now we can prove the proposition.

\begin{proof}[Proof of Proposition~\ref{prop:geom_mult}]
	For a subset $S \subseteq [2d+1]$ we use the shorthand $e_S := \prod_{i \in S} e_i$. Also, we denote by $ e_1\cdots e_{\hat{i}} \cdots e_{2d} $ the product where index $ i $ is omitted.  Note that
	\begin{align}\label{eq:Tdel}
		T_{\partial}&=e_{2d+1}(\sum_{i=1}^{2d} e_1\cdots e_{\hat{i}} \cdots e_{2d})\\
		&= \binom{2d}{d}^{-1} \sum_{i=1}^{2d} \sum_{\substack{S \subseteq [2d+1]\setminus\{i\} \\ |S|=d}} e_S \otimes e_{[2d+1]\setminus (S \cup \{i\})} \nonumber\\
		&=\binom{2d}{d}^{-1}  \sum_{\substack{S \subseteq [2d+1]\\ |S|=d}} e_S \otimes \left(\sum_{\substack{i \notin S \cup \{2d+1\}}} e_{[2d+1]\setminus(S \cup \{i\})}\right).\nonumber
	\end{align}
	As mentioned in \Cref{rem:generic-dim-of-L(b)}, the tensors $e_S$ are linearly independent, so the rank of $T_{\partial}$ is equal to the dimension of the subspace
	\begin{align*}
		\mathcal{S}:=\bigg\langle \sum_{\substack{i \notin S \cup \{2d+1\}}} e_{[2d+1]\setminus(S \cup \{i\})} : S \subseteq [2d+1], |S|=d \bigg\rangle.
	\end{align*}
	Note that we can split by whether $ S $ contains $ d+1 $ and we obtain
	\begin{align*}
		\mathcal{S}=\: &\bigg\langle \sum_{i \in Q} e_{Q \setminus \{i\}} : Q \subseteq [2d], |Q|=d+1 \bigg\rangle\\
		\oplus\: &\bigg\langle \sum_{i \in R} e_{(R \cup \{2d+1\}) \setminus \{i\}} : R \subseteq [2d], |R|=d \bigg\rangle
	\end{align*}
	We denote these two spaces by $\mathcal{Q}$ and $\mathcal{R}$, respectively. Note that $\mathcal{Q}$ and $\mathcal{R}$ are indeed in direct sum, as they are supported on disjoint sets of monomials. To complete the proof, we will show that $\dim(\mathcal{Q})=\dim(\mathcal{R})=\binom{2d}{d-1}$. We count the dimensions of these two vector spaces separately.
	\begin{claim}
		$\dim(\mathcal{Q})=\binom{2d}{d-1}$.
	\end{claim}
	Let $V_{k}:=\langle e_{S} : S \subseteq [2d], |S|=k \rangle$ for $ k \in \{d, d+1\} $, and define a linear map $D : V_{d+1} \rightarrow V_d$ by $D(e_Q)=\sum_{i \in Q} e_{Q \setminus \{i\}}$. Similarly, define $U: V_d \rightarrow V_{d+1}$ by $U(e_T)=\sum_{i \notin T} e_{T \cup \{i\}}$. Note that $\mathcal{Q}=D(V_{d+1})$. Since $\dim(V_{d+1})=\binom{2d}{d-1}$, it suffices to prove that $D$ is injective. We will do this by proving that $U\circ D: V_{d+1} \rightarrow V_{d+1}$ is invertible. Note that
	\begin{align*}
		UD(e_Q)=(d+1) e_Q + \sum_{\substack{S \subseteq [2d]\\|S|=d+1\\ |S\cap Q|=d}} e_S,
	\end{align*}
	so $UD=(d+1) I + A(2d,d+1)$. Recall that $ A(2d,d+1) $ denotes the adjacency matrix of the Johnson graph. By Lemma~\ref{lemma:johnson}, this map has eigenvalues
	\begin{align*}
		(d+1)+(d+1-j)(2d-(d+1)-j)-j=(d+1-j)(d-j)
	\end{align*}
	for $j=0,\dots, d-1$. Since these are all nonzero, $UD$ is invertible and the claim is proven.
	
	\begin{claim}
		$\dim(\mathcal{R})=\binom{2d}{d-1}$.
	\end{claim}
	\noindent
	To prove the claim, first note that $\mathcal{R}$ is isomorphic to the space
	\begin{align*}
		\mathcal{R}':=\bigg\langle \sum_{i \in R} e_{R \setminus \{i\}} : R \subseteq [2d], |R|=d \bigg \rangle,
	\end{align*}
	Define $D' : V_d \rightarrow V_{d-1}$ as $D'(e_R)=\sum_{i \in R} e_{R \setminus \{i\}}$ and $U': V_{d-1} \rightarrow V_d$ as $U'(e_S)=\sum_{i \notin S} e_{S \cup \{i\}}$. Arguing similarly to above, we have that $D'U'=(d+1)I+A(2d,d-1)$, which has eigenvalues $(d-j)(d+1-j)$ for $j=0,\dots, d-1$. Hence, $D'U'$ is invertible, so $U'$ is injective. Hence, $\dim(\mathcal{R}')=\rk(D') = \rk (U'D')$. Arguing similarly, we have $U'D'=d I + A(2d,d)$, which has eigenvalues $(d-j)(d-j+1)$ for $j=0,\dots, d$, and the zero eigenvalue has multiplicity $\binom{2d}{d}-\binom{2d}{d-1}$. This proves that $\rk(U'D')=\binom{2d}{d-1}$, and completes the proof.
\end{proof}

Finally, we need a criterion to group eigenvectors together, which belong to the same Chow rank 1 term. We will use a criterion similar to Step (4) in the algorithm of \Cref{thm:chow}, but this time it requires a bit more technical effort to prove. The following Proposition will lay the groundwork for the grouping criterion, by explicitly describing the eigenvectors. 

\begin{prop}\label{prop:grouping-test}
	Let $T=e_1 \cdots e_{2d+1}$, let $\partial=e_1+\dots + e_{2d}$ and $ \partial_{d+1}=e_1+\dots + e_{2d+1} $,  let $y=\prod_{j=1}^d (e_{2j}-e_{2j-1})$, let $\mu=e_1+\dots+e_{2d-1}+e_{2d+1}$, and let $z=(e_{2d+1}-e_{2d})\prod_{j=1}^{d-1} (e_{2j}-e_{2j-1})$. Then
	\begin{align}\label{eq:evector_check}
		T_{\partial} y =0,
	\end{align}
	and similarly,
	\begin{align}\label{eq:evector_check2}
		T_{\mu} z =0.
	\end{align}
	Furthermore,
	\begin{align}\label{eq:contraction_nonzero}
		z^T T_{\partial_{2d+1}} y \neq 0.
	\end{align}
\end{prop}
\begin{proof}
	Denote by $ \sigma(R) $ the number of odd elements in a set $ R $. Note that
	\begin{align*}
		y=\sum_{\substack{R \subseteq [2d],\: |R|=d\\ \forall  j\in [d]\colon |R \cap \{2j-1,2j\}|=1}} (-1)^{\sigma(R)} e_R.
	\end{align*}
	From here, it is easily calculated that for the contraction of $ T $ by $ y $, we have
	\begin{align*}
		(I_n^{\otimes d+1} \otimes y^{T})T= \alpha \; e_{2d+1} y
	\end{align*}
	for some non-zero scalar $\alpha \in \K^{\times}$. We first verify~\eqref{eq:evector_check}. Note that
	\begin{align*}
		T_{\partial} y = \alpha \sum_{\substack{R \subseteq [2d]\\ |R|=d\\ |R \cap \{2j-1,2j\}|=1 \forall  j\in [d]}} (-1)^{\sigma(R)} w_R,
	\end{align*}
	where
	\begin{align*}
		w_R=e_{2d+1}\sum_{i \in R} e_{R \setminus \{i\}}.
	\end{align*}
	
	Let $S \subseteq [2d]$ be a subset of size $d-1$. If $|S \cap \{2j,2j-1\}|=2$ for some $j \in [d]$, then none of the $w_R$ are supported on $S \cup \{2d+1\}$. Otherwise, the only possibility is that $|S \cap \{2j,2j-1\}|=1$ for all but one index $j' \in [d]$, and for this index $|S \cap \{2j',2j'-1\}|=0$. In this case, there are only two choices of $R$ for which $w_{R}$ is supported on $S \cup \{2d+1\}$, namely $R=S \cup \{2j'\}$ and $R=S \cup \{2j'-1\}$. The coefficients of $w_R$ for these two choices are additive inverses, hence $T_{\partial} y$ is not supported on $S \cup \{2d+1\}$. This verifies~\eqref{eq:evector_check}. By symmetry,~\eqref{eq:evector_check2} also holds.
	
	It remains to verify~\eqref{eq:contraction_nonzero}. Note that
	\begin{align*}
	T_{\partial_{2d+1}} y = \alpha y,
	\end{align*}
	so
	\begin{align*}
	z^T T_{\partial_{2d+1}} y=\alpha \; z^T y= \alpha \; 2^{d-1} (e_{2d+1}-e_{2d-1})^T (e_{2d}-e_{2d-1})= \alpha \; 2^{d-1} \neq 0.
	\end{align*}
	This completes the proof.
\end{proof}

Now for the algorithm. Let us write $ U_i = \langle a_{i,J} : |J| = d \rangle $ for the space spanned by the ``monomials'' of degree $ d $ in $ a_i = (a_{i1},\ldots,a_{id}) $. 
\begin{theorem}
	Suppose that $T=\sum_{i=1}^r a_{i1}\ldots a_{i(2d+1)}$ with $\{a_{ik}\}_{i\in [r],k=1,2,3}$ a set of points in general position. Denote $ U_i $ as described above. Then, the following algorithm produces the terms $a_{i1}\ldots a_{i(2d+1)}$ in time $ \mathcal{O}(n^{3d+1}) $, supposed that $r \le \dfrac{\binom{n+d-1}{d}}{\binom{2d+1}{d}} - \binom{2d+1}{d}$ and $ 2d+1\le n $. 
	\begin{enumerate}
		\item Sample generic $ f, g, \ell \in \K^n $ and compute the matrices $ T_f, T_g, T_{\ell} $ 
		\item By \Cref{thm:contraction-variety-higher-chow}, the generalized eigenvalue problem $ \det(T_f - \lambda T_g) = 0 $ has $ (2d+1)r $ solutions $ \lambda_{ik} $, with eigenspaces $ V_{ik} $ of dimension $ \binom{2d+1}{d}/(2d+1) $, by \Cref{prop:geom_mult}. Represent each eigenspace by a matrix $ A_{ik} $, whose columns form a basis for the eigenspace. 
		\item Set $ h_{ik} := f - \lambda_{ik} g $.  
		\item For each matrix $A_{i k}$, find the matrices $A_{i' k'}$ for which $A_{i' k'}^T T_{\ell} A_{i k} \neq 0$. This partitions the matrices into sets $\{A_{i 1},\ldots,A_{i(2d+1)}\}$.
		\item Compute $ \im T_\ell A_{i1}  + \ldots + \im T_\ell A_{i(2d+1)} $, which is equal to $ U_i $.
		\item Having computed bases for the spaces $U_i$, construct matrices $\Pi_i$ for which $\Pi_i |_{U_i}=I_{U_i}$ and $\Pi_i |_{U_j}=0$ for $j \neq i$. 
		\item Compute $ a_{ik}' = ((h_{ik}^{T})^{\otimes d-1} \otimes I_n) \Pi_i ((h_{ik}^{T})^{\otimes d+1} \otimes I_n^{\otimes d})T $.  Each $ a_{ik}' $ will be a multiple of $ a_{ik} $. 
		\item For each $ i=1,\ldots,r $, pick a general element $ x_{i1}\in V_{i1} $ and solve the system $T_{\ell} \; x_{i1}= c_i (a_{i1}' \cdots a_{i(2d+1)}')_{\ell} \; x_{i1}$ for $c_i \in \K$.
		\item Output $\{c_i a_{i1}' \cdots a_{i(2d+1)}' : i =1,\dots, r\}$.
	\end{enumerate}
\end{theorem}
\begin{proof}
	\textbf{Correctness:} \\
	Step (2): Covered by \Cref{thm:contraction-variety-higher-chow} and \Cref{prop:geom_mult}. \\
	Step (4): It is clear that $ A_{i' k'}^T T_{\ell} A_{i k} = 0 $ if $ i\ne i' $, since the vectors in the image of $ A_{ik} $ are orthogonal to any vector that is not in $ U_i $, and those in the image of $  A_{i' k'} $ are orthogonal to any vector not in $ U_{i'} $. 
	Therefore, it remains to show the converse direction, that is, $ A_{ik'}^T T_{\ell} A_{i k} \ne 0 $. By $ \mathfrak{S}_{2d+1} $ symmetry, it suffices to show this for the case $ k = d, k' = d+1 $. Also by $ \GL_n $ symmetry, we may assume that $ (a_{i1},\ldots,a_{i(2d+1)}) = (e_1,\ldots,e_{2d+1}) $.
	But then, the question corresponds exactly to the setting of \Cref{prop:grouping-test}: Indeed, the eigenspace $ V_{i(d+1)} $ corresponds to the kernel of the matrix $ T_{\partial} $ from \Cref{prop:grouping-test} and the eigenspace $ V_{id} $ corresponds to the kernel of the matrix $ T_{\mu} $ from \Cref{prop:grouping-test}. (Note that the $ T $ in \Cref{prop:grouping-test} equals our $ M_{i} $). 
	Now, \eqref{eq:contraction_nonzero} guarantees that there exist eigenvectors, i.e., elements $ z \in \im A_{id}$ and $ y\in \im A_{i(d+1)} $ such that $ z^{T} T_{\ell} y \ne 0  $ for the contraction $ \ell = e_1 + \ldots + e_{2d+1} $. In particular,  $ z^{T} T_{\ell} y$ will be nonzero for a generic contraction $ \ell $. This shows that    $ A_{id}^T T_{\ell} A_{i (d+1)} \ne 0 $.\\
	Step (5): We need to show that $ W_i := \im T_\ell A_{i1}  + \ldots + \im T_\ell A_{i(2d+1)} $ is indeed equal to $ U_i $. First, note that the spaces $ \im A_{i1},\ldots, \im A_{i(2d+1)} $ are eigenspaces of $ T_g^{-1}T_f $ to distinct eigenvalues, and must therefore be in direct sum. By \Cref{prop:geom_mult}, the sum of these images has dimension $ \binom{2d+1}{d} $. Since $ \rk(T_\ell) $ is maximal, it follows that $ W_i = T_\ell (\im A_{i1} + \ldots + \im A_{i(2d+1)}) $ also has dimension $ \binom{2d+1}{d} $. 
	Furthermore, $ \im T_\ell = \bigoplus_{i = 1}^r \im M_{if} $ and we have that $ \im A_{ik} \subseteq \ker M_{i'f} $ for each $ i'\ne i $. Thus, we have $ W_i\subseteq \im M_{if} \subseteq U_i $, and therefore $ U_i = W_i $ for dimension reasons.\\
	Step (7): Clearly, $ \Pi_i ((h_{ik}^{T})^{\otimes d+1} \otimes I_n^{\otimes d})T = ((h_{ik}^{T})^{\otimes d+1} \otimes I_n^{\otimes d})M_{i}  $, since $ \Pi_i $ annihilates every vector in $ U_{i'} $, if $ i'\ne i $, and $ \Pi_i $ acts as the identity on $ U_i $. Therefore, $  a_{ik}' = ((h_{ik}^{T})^{\otimes 2d} \otimes I_n) M_{i} $ is the $ 2d $-fold contraction of the Chow rank-1 tensor $ M_{i} = a_{i1}\cdots a_{i(2d+1)} $ by $ h_{ik} $. Recall that $ h_{ik} $ is orthogonal to $ a_{ik} $. Therefore, contracting with $ h_{ik} $ annihilates every term in the symmetrized tensor product $ a_{i1}\cdots a_{i(2d+1)} $ where $ a_{ik} $ is on the left side. Consequentially, if we contract $ 2d $ times by $ h_{ik} $, then the only terms that do not get annihilated are those, where $ a_{ik} $ is on the rightmost side. Each of those clearly yields a multiple of $ a_{ik} $. A short calculation shows that the multiple is nonzero, as it is proportional to $ \prod_{k' \ne k} \langle a_{ik'}, h \rangle $, and a general point $ h_{ik} $ in the contraction variety does not lie in the intersection of two irreducible components. \\
	Step (8): Note that $ Tx_{i1} = M_{i}x_{i1} $, since $ x_{i1} $ annihilates all $ M_{i'} $ with $ i'\ne i $. Since $ M_i $ and $ a_{i1}' \cdots a_{i(2d+1)}' $ are nonzero multiples of each other, it suffices to show that they do not vanish after contraction with $ x_{i1} $. But if $ 0 =  M_{i}x_{i1} =  T_{x_{i1}} $, then $ 0 = f^{T} T_{x_{i1}} = T_f x_{i1} $. This contradicts the fact that $ T_f $ is general and the eigenvectors $ x_{i1} $ lie in the image of a general $ T_f $, hence not in the kernel. 
	
	\textbf{Runtime: } We suppress constants of $ d $ in the $ \mathcal{O} $-notation. The claim is that the total runtime is $ \mathcal{O}(n^{3d+1}) $, with the dominating part being 	Step (7). All other steps take time $ \mathcal{O}(n^{3d}) $ or less.  \\
	Step (1): Computing the contractions and reshaping them into matrices $ T_f, T_g, T_{\ell} $ takes time $ \mathcal{O}(n^{2d+1}) $. \\
	Step (2): A generalized eigenvalue problem on matrices of size $ n^d $ can be solved in time $ \mathcal{O}(n^{3d}) $. \\
	Step (4): Computing the product $  T_{\ell} A_{i k} $ of an $ \mathcal{O}(n^d) \times \mathcal{O}(n^d) $ matrix with a matrix of size $ \mathcal{O}(n^d)\times \mathcal{O}(1) $ takes time $ \mathcal{O}(n^{2d}) $. Doing this for all eigenvectors takes time $ r\mathcal{O}(n^{2d}) $. Subsequently, computing the products of $ T_{\ell} A_{i k}  $ with all other $ A_{i' k'} $ takes time $ \mathcal{O}(n^{d}) $ for each $ i' $ and $ k' $, which results in a total time of $ \mathcal{O}(r^2n^d + rn^{2d}) $ for Step (4). Using $ r=\mathcal{O}( n^d )$, this takes time $ \mathcal{O}(n^{3d}) $. \\
	Step (5): The matrices $ T_{\ell} A_{ik} $ have already been computed in Step (4), so there is not much to do here. Joining their columns together yields a basis of $ \im T_{\ell} A_{i1}  + \ldots + \im T_{\ell} A_{i(2d+1)} $. This can be done in time $ \mathcal{O}(r) \le \mathcal{O}(n^d) $. \\
	Step (6): We already computed a basis $\mathcal{B} = \bigcup_{i = 1}^r \mathcal{B}_i $ of $ U_1\oplus \ldots \oplus U_r $ in Step (5), where each $ \mathcal{B}_i $ is a basis of $ U_i $. The matrices $ \Pi_i $ can be interpolated from a dual basis $ \mathcal{B}^{\ast} $ to $ \mathcal{B} $. The dual basis can be computed by matrix inversion, which on a vector space of dimension $ \mathcal{O}(n^d) $, which takes time $ \mathcal{O}(n^{3d}) $.  \\
	Step (7): Contracting $ T $ with a vector $ h_{ik} $ takes time $ \mathcal{O}(n^{2d+1}) $. All subsequent contractions are dominated in complexity by the first contraction. Thus we obtain the vector $ v := ((h_{ik}^{T})^{\otimes d+1} \otimes I_n^{\otimes d})T $ in time $ \mathcal{O}(n^{2d+1}) $. Multiplying the size $ \mathcal{O}(n^d)\times \mathcal{O}(n^d) $ matrix $ \Pi_i $ with the vector $ v $ takes time $ \mathcal{O}(n^{2d}) $ and is thus also bounded by the complexity of the first contraction. All subsequent contractions of $ \Pi_i v $ are of negligible complexity. We need to repeat the procedure for each $ i $ and $ k $, resulting in a complexity of $ \mathcal{O}(r n^{2d+1}) $. Using $ r=\mathcal{O}( n^d) $, we get a total complexity of $ \mathcal{O}(n^{3d+1}) $ for Step (7).	\\
	Step (8): The vectors $ T_{\ell} \; x_{i1} $ are linear combinations of the columns of $ T_{\ell}A_{i1} $, which were computed in Step (4). Therefore, there is not much additional effort to compute them. Contracting a Chow rank $ 1 $ matrix by $ \ell $ takes time $ \mathcal{O}(n) $, if the rank decomposition is known. The matrix vector multiplication of $ (a_{i1}' \cdots a_{i(2d+1)}')_{\ell} $ with $ x_{i1} $ takes time $ \mathcal{O}(n^{2d}) $. Thus doing Step (8) for each $ i = 1,\ldots,r $, takes at most time $ \mathcal{O}(rn^{2d}) \le \mathcal{O}(n^{3d})$. 
\end{proof}

	\section{Third order Chow decompositions into non-generic orbits}\label{sec:cubic-nonindependent}

Let $ T \in S^3(\K^n)$ be again a symmetric order-$ 3 $ tensor with a Chow decomposition
\begin{align*}
	T = \sum_{i = 1}^{r} a_{i1} a_{i2} a_{i3}
\end{align*}
Write $U_i=\langle a_{i1},a_{i2},a_{i3}\rangle$. In this section, we consider the case where $T$ is concise and the $U_i$ are in direct sum, i.e. $\K^n=\bigoplus_{i=1}^r U_i$. 
However, contrary to \Cref{sec:core} we {no longer} assume that the entire  $\{a_{ik}\}_{i\in [r],k=1,2,3}$ is linearly independent. Thus, linear dependencies are allowed among $ a_{i1}, a_{i2}, a_{i3} $. We call a tensor \emph{direct-sum concise}, if it is concise and it has a minimum Chow rank decomposition where the $ U_i $ are in direct sum.

If linear dependencies are allowed and $ n\ge 3 $, then each Chow-rank-1 term $ a_{i1} a_{i2} a_{i3} $ lies in one of $ 4 $ distinct $ \GL_n $-orbits, which are described as follows. 
\begin{enumerate}[(I)]
	\item The orbit $ \GL_n e_1^{\otimes 3}$ contains those Chow-rank-1 terms  $ abc $, where $ \dim \langle a,b,c \rangle = 1$. It equals the Veronese variety $ \nu_3(\mathbb K^n) $. 
	\item The generic orbit $ \GL_n e_1e_2e_3$ contains Chow-rank-1 terms $ abc $, where all $ a,b,c\in \mathbb K^n $ are linearly independent. This is the unique orbit where $\dim \langle a,b,c \rangle = 3$.  
	\item The orbit of the \emph{W-tensor} $ \GL_n e_1^2e_2$ contains those Chow-rank-1 terms  $ abc $, where either $ a $ and $ b $, or $ b $ and $ c $ are collinear (but not both). Here, $ \dim \langle a,b,c \rangle = 2 $. The orbit closure is the tangential variety to the Veronese variety $ \nu_3(\mathbb K^n) $. 
	\item The orbit $ \GL_n e_1e_2(e_1+e_2)$ contains those Chow-rank-1 terms  $ abc $, where $ \dim \langle a,b,c \rangle = 2$ but each pair of two vectors from $ a,b,c $ is linearly independent. Over $ \K = \C $, it equals the orbit of $ e_1^3 + e_2^3 $, since the bivariate polynomial $ x^3 + y^3 = (x+y)(x+\zeta y)(x+\zeta^2y) $ factors into a product of three linearly dependent linear forms, where $ \zeta \in \C $ denotes a primitive third root of unity. 
\end{enumerate}

Note that the complex factorization in case (IV) implies that the real and the complex case need different treatments. Indeed, over $ \K=\C $, the closure of orbit (IV) contains the rank-2 secant variety $ \sigma_2(\nu_3(\mathbb C^n)) $ of the Veronese variety. This means that already a sum of two general elements of orbit (IV) will not have a unique decomposition.\footnote{This is a special case of the fact that proper secant varieties can never be identifiable. Precisely, if $ X $ is any variety and $ Y = \sigma_2(X) $, then $ Y $ is not $ 2 $-identifiable (even when $ X $ is $ 4 $-identifiable). Cf. also the example in \cite[Remark 3.10]{Massarenti_Mella_2024}.} One may look at the example $T=e_1^3+e_2^3+e_3^3+e_4^3$, which has Chow rank 2.\footnote{The Chow rank of $T$ is at most $2$ by construction, and at least 2 because the flattening rank is $4$, whereas a Chow rank 1 tensor has flattening rank at most 3.} 
The tensor $ T $ is non-identifiable, as it can be written both as $ T = (e_1^3+e_2^3)+(e_3^3+e_4^3) $ and as $ T = (e_1^3+e_3^3)+(e_2^3+e_4^3) $, which give rise to two distinct Chow decompositions of minimum rank. Precisely, the decompositions are  
\begin{align*}
	T&=(e_1+e_2)(e_1 + \zeta e_2)(e_1+{\zeta}^2e_2)+(e_3+e_4)(e_3 + \zeta e_4)(e_3+\zeta^2 e_4),\\
	T&=(e_1+e_3)(e_1 + \zeta e_3)(e_1+\zeta^2 e_3)+(e_2+e_4)(e_2 + \zeta e_4)(e_2+\zeta^2 e_4).
\end{align*}
In addition, note that over $ \K=\C $, each point in the type (IV) orbit can be written as a sum of two points in the type I orbit. However, this is not a problem when considering the minimum rank decomposition. Note that if there is an odd number of type (I) orbits, then an arbitrary type (I) orbit need not be grouped into a pair. We now formulate the different objectives for $ \K=\C $ and $ \K = \R $.

\begin{question}[Real]\label{q:real}
	Given a symmetric, direct-sum concise tensor $ T \in \mathbb R^{n\times n\times n}$ of Chow rank $ n $, find the unique minimum Chow rank decomposition. 
\end{question}

\begin{question}[Complex]\label{q:complex}
	Given a symmetric, direct-sum concise tensor $ T \in \mathbb C^{n\times n\times n}$ of Chow rank $ n $, find a minimum Chow rank decomposition, and show uniqueness up to grouping pairs of type (I) orbits into type (IV) orbits.
\end{question}

Note that in the real case, it is possible to distinguish between a sum of two type (I) terms and a type (IV) term. A sum of two real type (I) terms will be of the form $ a^3 + b^3 $ with real $ a,b\in \R^n $, whereas a real type (IV) term can be written uniquely in the form $ z^3 + \bar{z}^3 $. The pairing into complex conjugates will ensure uniqueness of the minimum rank decomposition. Let us quickly prove this fact. 

\begin{lemma}
	A real Chow rank one tensor of type (IV) can be written uniquely as $x^3 + y^3 $ for vectors $x,y \in \C^n$. It furthermore holds that $x \notin \R^n$, and $y=\overline{x}$.
\end{lemma}
\begin{proof}
	Let $ a,b,c \in \mathbb R^n $ define a type IV term. That is, $ a,b,c $ are pairwise non-collinear and $ c\in \langle a,b \rangle $. Then, after a real change of coordinates $ \varphi \in \GL_n(\R) $, we may assume that $ a = e_1, b = e_2 $ and $ c = e_1+e_2 $. Consider the identity 
	\[
		e_1e_2(e_1+e_2) =\frac{1}{3(\zeta^2-\zeta)}\Big[(e_1-\zeta e_2)^{\otimes 3}-(e_1-\zeta^2 e_2)^{\otimes 3}\Big], 
	\]
	where $ \zeta \in \C $ is a primitive third root of unity. This shows that $ e_1e_2(e_1+e_2) = z^3 + \bar{z}^3 $ for $ z = e_1-\zeta e_2 $. Note that $ \zeta^2 = \bar{\zeta} $. 
	After a reversing the real linear change of coordinates, $ abc = (\varphi^{-1}(z))^{\otimes 3} + (\overline{\varphi^{-1}(z)})^{\otimes 3} $ is still of the desired form. By construction, $\varphi^{-1}(z)\notin \R^n$, because $abc$ does not have Waring rank 1.
\end{proof}
Prony's method and our main algorithm can be used to decompose symmetric tensors into class (I) and (II) orbits, respectively. We first present an algorithm for decomposing a tensor into class (III) orbits (W-tensors), and then describe how the three different algorithms for classes (I-III) can be combined into a single Chow decomposition algorithm.

\subsection{Decomposition into W-tensors}\label{sec:W}

Let $T=\sum_{i=1}^r a_{i1}^2 a_{i2}$. In this subsection we give an $\mathcal{O}(n^4)$-time algorithm for recovering these Chow rank 1 terms when the spaces $U_i:=\langle a_{i1}, a_{i2} \rangle$ are in direct sum. We first determine the contraction variety.

\begin{prop}
It holds that $X_T=\bigcup_{i=1}^r \langle a_{i1} \rangle^{\perp}$.
\end{prop}
\begin{proof}
	Let $M_i=a_{i1}^2 a_{i2}$. For $ f\in \K^n $, we may write $ T_f = \sum_{i = 1}^{r} M_{if} $, where $ M_{if} $ is the contraction of $ a_{i1}^2 a_{i2} $ along $ f^T$. Note that for generic $f$, $\im M_{if} = U_i$ for all $i$. Note also that $\text{rank}(T_f)=2r$ for generic $f\in \K^n$, and the rank of $T_f$ drops (i.e., $f$ lies in the contraction variety) if and only if $\text{rank}(M_{if}) < 2$ for some $i$. For such $ i $, we can assume without loss of generality that $a_{i1}=e_1$ and $a_{i2}=e_2$. Let $\alpha_1, \alpha_2$ be the first two coordinates of $f$. Then
	\begin{align}\label{eq:w-state-geometric-multiplicity}
		3 M_{i f} = 2 \alpha_1 e_1 e_2 + \alpha_2 e_1 e_1,
	\end{align}
	which has rank $\leq 1$ if and only if $\alpha_1=0$. This completes the proof.
\end{proof}

Hence, the contraction variety is a union of $r$ distinct hyperplanes. 
Furthermore, we see that for generic $f$ and $g$ all of the generalized eigenvalues will have geometric multiplicity 1 (and algebraic multiplicity 2), as these correspond to intersection points of the generic line $f-\lambda g$ with the contraction variety.

Our decomposition algorithm now proceeds as follows. First, we choose generic $f,g \in \K^n$, and solve the generalized eigenvalue problem
\begin{align*}
	\det(T_f-\lambda T_g)=0.
\end{align*}
The solutions $\lambda$ are precisely those for which the rank of $M_{i (f-\lambda g)}$ is 1 for some $i$. We denote these eigenvalues by $\lambda_{i}$. All are distinct with geometric multiplicity 1.

We let $a_{i1}'=T_{f}x_{i1}$, which is a nonzero scalar multiple of $a_{i1}$. Next, we let $b_i$ be a vector that is orthogonal to $a_{j1}$ for all $j\neq i$, and satisfies $\langle b_i, a_{i1}\rangle =1$. Then $3((b_i^T)^{\otimes 2} \otimes I)T= a_{i2} + 2 \langle b_i, a_{i2} \rangle a_{i1}$. This gives a second basis element $v_i$ for $U_i$.

Now, to recover the terms $a_{i1}^2 a_{i2}$, we let $\Pi_i$ be the matrix for which $\Pi_i |_{U_i}= I_{U_i}$ and $\Pi_i |_{U_j} =0$ for $j \neq i$. Then
\begin{align*}
\Pi_i T_f g &= (a_{i1}^2 a_{i2})_f g = ((a_{i1}')^2(\alpha a_{i1}'+\beta v_i))_f g\\
&= \alpha \frac{1}{3} \langle f, a_{i1}' \rangle \langle g, a_{i1}' \rangle v_i \\
&\quad+ \beta \frac{1}{3}(\langle f, a_{i1}' \rangle \langle g, a_{i1}' \rangle v_i + (\langle f,a_{i1}' \rangle \langle g, v_i \rangle + \langle f, v_i \rangle \langle g, a_{i1}' \rangle) a_{i1}')
\end{align*}
This gives a linear system of equations for $\alpha$ and $\beta$, which uniquely determines $\alpha$ and $\beta$ because they appear as the coefficients of linearly independent vectors. It takes $\mathcal{O}(n^3)$ time to compute all the matrices $\Pi_i$ collectively, and $\mathcal{O}(n^2)$ time for each index to solve the linear system for $\alpha$ and $\beta$.

We have proven the following theorem:

\begin{theorem}\label{thm:general_chow}
	Suppose that $T=\sum_{i=1}^r a_{i1}^2 a_{i2}$ with $U_i := \langle a_{i1}, a_{i2}\rangle$ in direct sum. Then this is the unique minimal decomposition of $T$ into W-tensors, and the following algorithm produces the terms $a_{i1}^2 a_{i2}$:
	\begin{enumerate}
		\item Sample generic $ f, g, \ell \in \K^n $. 
		\item Find the solutions $ \lambda_{i}, x_{i} $ to $ \det(T_f - \lambda T_g) = 0 $ (generalized eigenvalues and eigenvectors).
		\item Let $a_{i1}'=T_f x_{i1}$, and let $b_i$ be a vector that is perpendicular to all $a_{j1}'$ for $j \neq i$, and satisfies $\langle b_i, a_{i1} \rangle =1$. Compute $ v_i = ((b_i^T)^{\otimes 2} \otimes I)T  $ and let $U_i=\langle a_{i1}', v_i\rangle $.	
		\item Let $\Pi_i$ be the matrix for which $\Pi_i |_{U_i}= I_{U_i}$ and $\Pi_i |_{U_j} =0$ for $j \neq i$.
		\item Solve the linear system $\Pi_i T_f g= ((a_{i1}')^2(\alpha a_{i1}'+\beta v_i))_f g$ for $\alpha$ and $\beta$. Let $a_{i2}'=\alpha a_{i1}'+\beta v_i$, so that $a_{i1}^2a_{i2}=(a_{i1}')^2 a_{i2}'$.
		\item Output $\{(a_{i1}')^2 a_{i2}' : i=1,\dots, r\}$.
	\end{enumerate}
\end{theorem}

\subsection{Runtime}
This algorithm has runtime $\mathcal{O}(n^4)$. Step (1) takes time $\mathcal{O}(n)$. Step (2) takes time $\mathcal{O}(n^3)$. Step (3) takes time $\mathcal{O}(n^3)$ for each index $i$, hence $\mathcal{O}(n^4)$ total. Steps (4-6) take time $\mathcal{O}(n^3)$.

\subsection{Third order Chow decompositions into mixtures of orbits}

Prony's method and our main algorithm decompose symmetric tensors into type (I) and (II) orbits, respectively, and our algorithm in the previous subsection can decompose symmetric tensors into type (III) orbits. In this subsection we show how to combine these algorithms to answer Questions~\ref{q:real} (real case) and~\ref{q:complex} (complex case).

Let $T=\sum_{i=1}^{r} a_{i1} a_{i2} a_{i3}$ be a direct-sum-concise tensor. We will present an $\mathcal{O}(n^4)$-time algorithm in both the real and complex settings.
We first determine the contraction variety by analyzing the four cases separately. 
\begin{enumerate}[(I)]
	\item In this case, $U_i = \langle a_{i1} \rangle$ and $\text{rank}(M_{if})$ drops if and only if $f \in \langle a_{i1} \rangle^{\perp}$.
	\item In this case, $a_{i1},a_{i2},a_{i3}$ are linearly independent, and $\text{rank}(M_{if})$ drops if and only if $f \in \langle a_{i1} \rangle^{\perp} \cup \langle a_{i2}\rangle^{\perp} \cup \langle a_{i3} \rangle^{\perp}$. This was proven in \Cref{sec:core}. 
	\item In this case, we have $\dim U_i=2$ and, without loss of generality, $M_i=a_{i1}^2 a_{i2}$, with $a_{i2} \notin \langle a_{i1}\rangle$. In this case, $\text{rank}(M_{if}) <2$ if and only if $f \in \langle a_{i1}\rangle^{\perp}$. This was proven in \Cref{sec:W}.
	\item In this case, there exist complex linearly independent vectors $v$ and $w$ for which $M_i=v^3+w^3$. The rank drops if and only if $f \in \langle v\rangle^{\perp} \cup \langle w \rangle^{\perp}$.
\end{enumerate}

So, by intersecting the contraction variety $ X_T $ with a general line, we obtain one eigenvector for each index in cases (I) and (III), three eigenvectors for each index in case (II), and two eigenvectors for each index in case (IV). 
Our algorithm proceeds then in two steps:
\begin{enumerate}
	\item Classify, which eigenvectors belong to types (II) and (III), and partition them according to the index $i$.
	\item For the remaining eigenvectors:
		\begin{enumerate}
			\item For $\K=\R$: Classify the remaining eigenvectors into types (I) and (IV) and partition them according to the index $i$.
			\item For $\K=\C$: Pair up the eigenvectors arbitrarily (giving rise to type (IV) terms in the decomposition), possibly with one leftover (giving rise to a type (I) term in the decomposition).
		\end{enumerate}
	\item Recover a minimum Chow rank decomposition of $T$.
\end{enumerate}

\subsubsection{Step 1: Partition the eigenvectors along $i$ for types (II) and (III)}

For the first step, we can partition the eigenvectors corresponding to indices in type (II) using a similar trick as in Section~\ref{sec:core}:
\begin{prop}\label{prop:trick-for-dim-Ui-3}
Let $x_{ij}, x_{i' j'}, x_{i'' j''}$ be distinct eigenvectors. Then the three contractions $x_{i j}^T T_{\ell} x_{i' j'}$, $x_{i j}^T T_{\ell} x_{i'' j''}$, $x_{i' j'}^T T_{\ell} x_{i'' j''}$ are all nonzero if and only if $\dim(U_i)=3$ and $i=i'=i''$.
\end{prop}
\begin{proof}
This is proven essentially identically to \Cref{prop:source-trick-for-dim-Ui-3}.
\end{proof}

For the remaining eigenvectors, we can determine if they belong to case (III) as follows:

\begin{prop}\label{prop:special_case}
	Let $x_{ij}$ be a generalized eigenvector of $(T_f,T_g)$, and let $\ell \in \K^n$ be generic. Then $x_{ij}^T T_{\ell} x_{ij}=0$ if and only if $i$ belongs to case (III).
\end{prop}
\begin{proof}
	First suppose that $i$ belongs to case (III), i.e. $\dim(U_i)=2$ and $a_{i1},a_{i2},a_{i3}$ are not pairwise linearly independent. In this case, there is a single generalized eigenvalue $\lambda_{i1}$ with geometric multiplicity $1$. Without loss of generality we can assume $M_i=e_1^2 e_2$. Note that $3! M_{i e_2}= e_1 e_1^T$. Hence, $M_{if} e_2 = \langle f, e_1 \rangle  e_1$ and $M_{ig} e_2= \langle g, e_1 \rangle e_1$  are both non-zero scalar multiples of $e_1$. It follows that $e_2$ is a generalized eigenvector of $(M_{if},M_{ig})$. Since $T_{x_{ij}}x_{ij}= (M_i)_{e_2}e_2=0$, this direction is proven.
	
	Conversely, suppose that $x_{ij}^T T_{\ell}x_{ij}=0$ for generic $\ell$. Then $T_{x_{ij}} x_{ij}=0$. If $i$ belongs to type (I), then $T_{x_{ij}}x_{ij}$ is a nonzero scalar multiple of $a_{i1}$, hence nonzero. A similar argument applies for type (IV). The index $i$ can also not have type (II) by Proposition~\ref{prop:trick-for-dim-Ui-3}. Hence, $i$ has type (III).
\end{proof}

\subsubsection{Step 2: Partition the remaining eigenvectors}
At this point, only type (I) and (IV) eigenvectors remain to be sorted. Over the real numbers, the eigenvectors $x_{i1}, x_{i2}$ in a type (IV) orbit will satisfy the property that $T_{x_{i1}} x_{i1}$ is a scalar multiple of $\overline{T_{x_{i2}} x_{i2}}$. In this way, we pair up the type (IV) indices uniquely, and classify the remaining indices as type (I).

Over the complex numbers, due to the non-uniqueness described above, an arbitrary pairing of these eigenvectors will correspond to a sum of two type (I) orbits, i.e. a single type (IV) orbit. This arbitrary pairing leaves possibly one eigenvector leftover, which will correspond to a single type (I) orbit. By construction, the decomposition will be unique up to this choice.

\subsubsection{Step 3: Recover a minimum Chow rank decomposition of $T$.}

This step is almost as straightforward as running the algorithms for each type independently.

\begin{enumerate}[(I)]
\item In this case, we can set $a_{i1}'$ to be any non-zero element of the one-dimensional space $U_i$, solve the system $T_f x_{i1}=c \langle f, a_{i1}'\rangle \langle x_{i1}, a_{i1}'\rangle a_{i1}'$ for $\alpha \in \C$, and let $a_{i1}^3=\alpha (a_{i1}')^3$. This takes time $\mathcal{O}(n^2)$ for each index $i$.
\item In this case we can recover the term $a_{i1}a_{i2}a_{i3}$ identically as in Section~\ref{sec:core}. This takes time $\mathcal{O}(n^2)$ for each index $i$.
\item In this case, we need a small modification of the W-tensor decomposition algorithm to recover the spaces $U_i$. For each index $i$ of type (III), let $a_{i1}'=T_f x_{i}$. Let $b_i$ be a vector that is orthogonal to all spaces $U_i$ from cases (I,II, IV), orthogonal to all vectors $a_{j1}'$ for $j \neq i$, and satisfies $\langle b_i, a_{j1}' \rangle =1$. Then $U_i=\langle a_{i1}',((b_i^T)^{\otimes 2} \otimes I)T \rangle $. From here, we can recover the components $a_{i1}^2a_{i2}$ identically as in Section~\ref{sec:W}. In this case, it takes $\mathcal{O}(n^4)$ time to construct the spaces $U_i$. It takes $\mathcal{O}(n^3)$ time to compute all the matrices $\Pi_i$ collectively, and $\mathcal{O}(n^2)$ time for each index to solve the linear system for $\alpha$ and $\beta$ for each $i$, so $\mathcal{O}(n^3)$ time total. Hence, this step takes time $\mathcal{O}(n^4)$ time total.

\item In this case, for the two eigenvectors $x, y$ we let $v'=T_{x} x$, let $w'=T_y y$, and solve the systems $T_f x=\alpha \langle f, v' \rangle \langle x, v' \rangle v'$ and $T_f y=\beta \langle f, w' \rangle \langle x, w' \rangle w'$ for $\alpha,\beta \in \C$. We then let $v:=\alpha^{1/3} v'$ and $w:=\beta^{1/3} w'$, and determine vectors $a,b,c$ for which $v^3+w^3=abc$. 
The vectors $ a,b,c $ are given by $ a = v+w $, $ b = v+\zeta w $ and $ c = v+\zeta^2 w $, where $ \zeta $ is a primitive third root of unity.  
In total, the procedure takes time $\mathcal{O}(n^2)$ for each pair of eigenvectors.
\end{enumerate}

We have proven the following theorem. We omit the step-by-step version for brevity.

\begin{theorem}
Let $T=\sum_{i=1}^r a_{i1} a_{i2} a_{i3}$ be a direct-sum concise symmetric tensor. Then there is an $\mathcal{O}(n^4)$-time algorithm to recover a minimum Chow rank decomposition of $T$. This decomposition is unique if $\K=\R$, and if $\K=\C$ it is unique up to a choice of grouping pairs of type (I) terms into type (IV) terms.
\end{theorem}

	\subsection*{Acknowledgements}
	Alexander Taveira Blomenhofer is supported by the ERC grant of Matthias Christandl under Agreement 818761. Part of this work was done while Benjamin Lovitz was supported by the National Science Foundation under Award No. DMS-2202782.  We thank Matthias Christandl for inviting and funding a visit of Benjamin Lovitz at the QMATH centre of the University of Copenhagen. We also thank Eric Evert, Tim Seynnaeve, Kevin Stubbs, Daniele Taufer, Nick Vannieuwenhoven, and Aravindan Vijayaraghavan for valuable discussions. 

	\bibliography{bibML}
	\bibliographystyle{plain}

\end{document}